\documentclass[APS]{revtex4}
\pdfoutput=1
\usepackage{amsfonts} 
\usepackage{graphicx}  
\usepackage{amsmath}\usepackage{amsthm}

\def\idty{{\leavevmode\rm 1\mkern -5.4mu I}} 
\def\Rl{{\mathbb R}}\def\Cx{{\mathbb C}}     
\def\norm #1{\Vert #1\Vert}

\def\ket #1{\vert #1\rangle}

\def\tr{\mathop{\rm Tr}\nolimits}

\mathchardef\ree="023C \mathchardef\imm="023D  

\def\BB{{\mathcal B}}
\def\HH{{\mathcal H}}
\def\KK{{\mathcal K}}
\def\AA{{\mathcal A}}

\def\deltahalf{\delta_{\frac12,\frac12}}

\newtheorem{thm}{Theorem}

\newtheorem{prop}[thm]{Proposition}

\begin{document}

\title{Maximal violation of Bell inequalities by position measurements}

\author{J. Kiukas}
\email{jukka.kiukas@itp.uni-hannover.de}
\author{R.F. Werner }
\email{reinhard.werner@itp.uni-hannover.de}

\affiliation{Inst. Theoret. Physik, Leibniz Universit\"at Hannover, Appelstr. 2 , 30167 Hannover, Germany}

\begin{abstract}
We show that it is possible to find maximal violations of the CHSH-Bell inequality using only position measurements on a pair of entangled non-relativistic free  particles. The device settings required in the CHSH inequality are done by choosing one of two times at which position is measured. For different assignments of the "+" outcome to positions, namely to an interval, to a half line, or to a periodic set, we determine violations of the inequalities, and states where they are attained. These results have consequences for the hidden variable theories of Bohm and Nelson, in which the two-time correlations between distant particle trajectories have a joint distribution, and hence cannot violate any Bell inequality.
\end{abstract}

\maketitle

\section{Introduction}

It is well-known that the position operators of a particle at different times do not in general commute. This is the reason why the notion of trajectories cannot be applied to quantum particles. But non-commutativity is also a useful feature in some experiments. In particular, it is essential in experiments aiming at violations of Bell inequalities. In this paper we show that the non-commutativity of positions at different times is sufficient for getting a maximal violation of the CHSH-Bell inequality, and find the states required for this.

Our first motivation for investigating this was the possibility of using such Bell experiments as a refutation of Hidden Variable Theories which do assign a distribution of trajectories to every quantum state: In such a theory the positions at all times have a joint distribution, and therefore cannot violate a Bell inequality. Hence their predictions must be in conflict with quantum mechanics and, most likely, with experiment. After completion of our work we found that this line of reasoning had already been followed by Correggi and Morchio \cite{Correggi}. We nevertheless include our discussion, and emphasize some additional points. Technically, the Bell violations found in \cite{Correggi} are for particles in an external potential, whereas we look at two free particles.

Our second motivation for the present work is the endeavour of finding a loophole free Bell test based on homodyne detection in quantum optics. In such a Bell measurement, each detection must be a function of just a single field quadrature, which is mathematically the same problem as using functions of a single position variable. This is impossible with Gaussian states, because the Wigner function then provides a joint distribution. But with the new abundance of non-Gaussian states recently realized in the lab \cite{Grangier,Auberson} there is a chance to find a feasible setup. Here the knowledge of the maximally violating states may be helpful, although only as a rough indication where to look. It would be even better to be able to start from a given state, and to identify the quadrature measurements giving the best violation.

Our paper and our results are organized as follows:
In Sect.~\ref{sec:Bohm} we briefly describe how our result contributes to the debate about hidden variable theories. In Sect.~\ref{sec:CHSH} we provide some general background concerning violations of the CHSH inequality. Here we introduce techniques related to the universal C*-algebra generated by two projections. These techniques are well known in the operator algebra literature, but as far as we know they have not been applied to simplify the theory of Bell inequality violations. In Sect.~\ref{sec:position} we outline how to get maximal CHSH violations from position measurements. There are three different settings: (1) We choose the ``+''-event of each measurement as a position outcome in some finite interval. When for both Alice and Bob, $d_1,d_2$ are the interval lengths used for the first and second setting, $m$ is the particle mass, and $t$ the time separation, then the attainable violation depends only on the dimensionless parameter $u=md_1d_2/(4t\hbar)$. We show that maximal violation is attained for infinitely many values of $u$ and also in the limit $u\to\infty$. (2) When the ``+''-event means that the particle is on the positive half-line, maximal violation can be almost achieved, up to an arbitrarily small error. Hence there are singular states (i.e., states not given by a density operator) for which maximal violation is attained. These are necessarily dilation invariant, up to a quadratic phase. Finally, (3) we look at periodic sets. It is well-known \cite{Busch}, that periodic functions of position and momentum commute, if the product of the periods is $2\pi \hbar / \text{integer}$. Translated to the setting of a free particle with time difference $t$ between position measurements with period $p_1$ and $p_2$, we find the two measurements to commute whenever $u^{-1}\in \mathbb{Z}$, where $u=mp_1p_2/(2\pi t\hbar)$. Of course, in that case no violation of a Bell inequality is possible. However, we show that this situation is very unstable, i.e., that the maximal violation jumps from zero to a finite value for $u$ arbitrarily close to an integer.
In an appendix we have collected some technicalities regarding case (2).

\section{The Bohm-Nelson theory}
\label{sec:Bohm}

When it first appeared, Bohm's hidden variable extension of quantum mechanics \cite{Bohm} met much opposition from the mainstream physical community because it appeared to violate some basic tenets of quantum theory. Heisenberg, in his paper introducing the uncertainty relations, had convinced the physics community that the notion of trajectories of individual particles had no place in the theory. There was even a theorem by von Neumann showing hidden variables to be impossible. To its proponents, Bohm was seen to restore Realism to physics, giving a complete moment-to-moment account of where all the particles of a complicated quantum system really were. In part, these were also the motivations of Edward Nelson for creating his ``stochastic mechanics''\cite{Nelson}. In addition, he endeavored to give a derivation of the theory, which was at the same time a derivation of quantum mechanics itself. Both theories are embedded in a family of such theories parameterized by the diffusion constant in units of $\hbar$ \cite{Davidson}, with Bohm's theory corresponding to $0$ and Nelson's to $1$. In the limiting case of infinite diffusion constant, we find a theory in which positions at different times are just taken to be independent%
\nocite{genRFW}%
\footnote{There are further generalizations. In fact, it is easy to build Markov processes following the quantum evolution of any observable \cite{genRFW}, so one could directly discuss everything in terms of spin variables rather than positions. However, this generalization is clearly against the taste of Bohmians, and we will ignore it as they have.}%
. For the conceptual problems we discuss here, the only salient feature of all these theories is that they provide a joint distribution for all particle positions at all times, such that the equal time probabilities for particle configurations exactly reproduce the quantum mechanical probability distribution $|\psi|^2$. This is the basis for the claim that Bohm's theory is empirically equivalent to quantum theory. At the very least, this agreement reassures us that some aspect of these ``real'' trajectories is correct.

On exactly the same footing, let us look at another quantity, which makes sense quantum and hidden variable  theories alike, namely the two-time correlations between distant non-interacting, but possibly entangled, subsystems. Of course, in the quantum case the positions of the same particle at different times do not commute, so quantum mechanics has no joint probability for these. But for correlations between different particles there is no such constraint, and we can directly compare the quantum predictions with the two-time correlation functions from the Bohm-Nelson theory. As we show below, the quantum and the Bohm-Nelson predictions turn out to be quite different (they also disagree between the Bohmian and the Nelsonian variants). So if we take the agreement of one-time correlations with quantum theory as evidence that there is something right about these trajectories, we are now forced to admit that there is also something wrong with them. Certainly, this disagreement completely invalidates the argument of ``empirical equivalence'' between Nelson-Bohm theory and quantum mechanics. We could even stage an experimentum crucis on the basis of the explicit states and observables computed below. There is little doubt how these would turn out, probably not even for the staunch defenders of these theories. So our argument in some sense refutes the Bohm-Nelson theory.

Of course, we are aware that the Bohmians and Nelsonians know about this disagreement, and will not be impressed \footnote{A notable exception is the founder of stochastic mechanics, who abandoned the theory, when he realized some of its unphysical non-local features \cite{Ascona}}. The simplest position is to include the collapse of the wave function into the theory \cite{Blanchard,Morato}. Then the first measurement instantaneously collapses the wave function.  So if agreement with quantum mechanics is to be kept, the probability distribution changes suddenly. There is no way to fit this with continuous trajectories: When the guiding field collapses, the particles must jump. While the glaring non-locality of this process may be seen as just another instance of implicate order, it introduces an element of unexplained randomness, and demotes the Bohm equation (or Nelson's Fokker-Planck equation) from its role as the fundamental dynamical equation for position.

This may be the reasons why many Bohmians  adopt a strongly contextual view. In this view one has to describe the measurement devices explicitly in the same theory, so all trajectories depend on the entire experimental arrangement. Therefore the trajectory probabilities in two experiments, in which the measurements on particle A happen at different times, have no relation to each other, not even for trajectories of particle B. So the two-time correlations computed from the 2-particle ensemble of trajectories are never observed anyhow, and hence pose no threat to the theory. The downside of this argument is that it also applies to single time measurements, i.e., the agreement between Bohm-Nelson configurational probabilities and quantum ones is equally irrelevant. The naive version of Bohmian theory holds ``position'' to be special, even ``real'', while all other measurement outcomes can only be described indirectly by including the measurement devices. Saving the Nelson-Bohm theory's failure regarding two-time two-particle correlations by going contextual also for position just means that the particle positions are declared {\it unobservable according to the theory itself}, hence truly hidden.

In this consistently contextual version of the theory, there may still be those ``real'' trajectories, but they are only for the eyes of Bohm's Demon, or some such hypothetical creature. No physical interaction, not even an ``ideal position measurement'', will reveal them to the mere human. This certainly explains the apparent paradox that Bohmians on the one hand place so much value on being able to say where the particles really are, but are, on the other hand, so remarkably uninterested in actually computing trajectories. But when the interest in the real trajectories is gone, the only gain from the whole theory seems to be a pro forma justification for saying that the hand of a voltmeter is really somewhere. The mountain in labor gave birth to a mouse.

\section{General structure of CHSH violations}
\label{sec:CHSH}
In this chapter we look at the general problem of finding the maximal quantum violations of the Bell-CHSH inequality, when the measurements of Alice and Bob are given. All this is well-known, but since we need it several times, it may be useful to state the criteria in a compact form.

Each of the measurements in the CHSH setting is a POVM with outcomes ${+1,-1}$, which means that it is characterized by two positive operators $F_\pm$ with $F_++F_-=\idty$. We can parameterize such observables by the single operator $A=F_+-F_-$, which gives the expectation of the outcome, and satisfies $-\idty\leq A\leq\idty$. Then $F_\pm=(\idty\pm A)/2$. We will assume the measurement to be projection valued (i.e. $F_\pm^2=F_\pm$), which is equivalent to $A=A^*$ and $A^2=\idty$. In the CHSH setting Alice chooses two such measurements, $A_1,A_2$, and Bob chooses $B_1,B_2$. Since their respective labs are widely separated, they can make their choices independently, and we may take $A$ and $B$ as acting on the respective tensor factors of the Hilbert space $\HH_A\otimes \HH_B$ associated with the combined system. There is a quantum state $\rho$ of this system, in which correlations $\tr\rho A_iB_j$ can be determined. The {\it CHSH-correlation} is the linear combination of four such terms, which is the expectation of the operator
\begin{equation}\label{bellop}
    T=A_1\otimes (B_1+B_2)+A_2\otimes (B_1-B_2).
\end{equation}
Given the operators $A_i,B_j$, the supremum of the CHSH-correlations $\tr\rho T$ attainable with quantum states $\rho$ are given by the largest elements in the spectrum of $T$, and since we can invert the outcomes of Alice's measurements ($A_i\mapsto-A_i$) we are equally interested in the most negative expectations. To summarize, we are looking for the operator norm $\norm T$. For arbitrary operators $X$ we have $\norm X^2=\norm{X^*X}$, and since $T=T^*$ we have $\norm T=\sqrt{\norm{T^2}}$. A simple algebraic computation using the properties of the operators $A_i,B_j$ stated above gives
\begin{equation}\label{bellopsq}
    T^2=4(\idty+A_3\otimes B_3),
\end{equation}
where we have denoted e.g. $A_3:= (2i)^{-1}[A_1,A_2]$. Since $A_1^2=\idty$, $A_1$ is unitary, and we have $A_1A_3A_1=-A_3$. Hence, the spectrum of $A_3$ is symmetrical around zero, with maximum equal to $\norm{A_3}$. This gives a compact expression for the maximal attainable CHSH-correlation, namely
\begin{equation}\label{maxviol}
    \norm T=\sqrt{4(1+\norm{A_3}\,\norm{B_3})}.
\end{equation}
In particular, when either Alice's or Bob's measurements commute, so the norm of the corresponding commutator vanishes, we get $\norm T=2$, i.e., the CHSH inequality is satisfied. On the other hand, since
$\norm{[A_1,A_2]}\leq 2\norm{A_1}\norm{A_2}=2$, we have $\norm T\leq\sqrt8=2\sqrt2$, which is known as Tsirelson's inequality. For our purposes, the main gain from (\ref{maxviol}) is that the determination of the maximal violation is reduced to the estimates of commutators, which can be done separately for Alice and Bob.

\subsection{The algebra generated by two projections}
Note that both on Alice's side (and similarly on Bob's) only two projections $P_i=(\idty+A_i)/2$ are relevant. Let $\AA(P_1,P_2)$ denote the algebra generated by the two projections $P_1$ and $P_2$ together with the identity $\idty$. It turns out \cite{Halmos,Davis,Raeburn,Borat} that this can be understood completely in terms of $2\times2$-matrices.  Indeed, we observe that
\begin{equation}\label{centerP2}
    C=\idty-P_1-P_2+P_1P_2+P_2P_1
\end{equation}
satisfies
\begin{equation}\label{centerP2p}
    CP_1=P_1P_2P_1=P_1C,
\end{equation}
and a similar relation for $P_2$. Therefore $C$ commutes with the generating projections of the algebra, and hence with all of $\AA(P_1,P_2)$. The central element $C$ satisfies $0\leq C\leq\idty$, because
$C=\idty-(P_1-P_2)^2$, and
\begin{equation*}
 -\idty\leq -P_2\leq   (P_1-P_2)\leq P_1\leq\idty,
\end{equation*}
so $0\leq(P_1-P_2)^2\leq\idty$. Clearly, $C=0$ means that $P_1$ and $P_2$ are orthogonal, whereas $C=\idty$ means that $P_1$ and $P_2$ are equal. At these extremes, $[P_1,P_2]=0$. More generally, the commutator satisfies
\begin{equation}\label{commPP2}
    [P_1,P_2]^*[P_1,P_2]=-[P_1,P_2]^2=C(\idty-C).
\end{equation}
Hence the largest norm for the commutator square is $\frac14$, attained at $C=\frac12\idty$, where indeed the operators
$A_i=2P_i-\idty$ appearing in the CHSH-inequality also attain their maximal commutator norm
\begin{equation}\label{commPP2n}
    \norm{A_3} = \tfrac 12\norm{[A_1,A_2]}=2\norm{[P_1,P_2]}=2 \sqrt{\norm{C(\idty-C)}}=1.
\end{equation}

Now we can express every element of $\AA(P_1,P_2)$ as a linear combination of the four terms $\idty, P_1, P_2, P_1P_2$, multiplied by suitable polynomials in $C$. It is convenient to choose another basis, in which multiplication becomes ordinary matrix multiplication, and
such an isomorphism can be implemented at the Hilbert space level (see e.g. \cite{Halmos,Davis}):
First we split off the null space of all commutators, i.e., $\HH_0=\{\phi|C(\idty-C)\phi=0\} = \ker A_3$. On this space, which is clearly an invariant subspace of the two projections, all four eigenvalue combinations of two commuting projections are possible. Going to the orthogonal complement $\HH_0^\perp$, we put $\KK:=\{\phi\in\HH_0^\perp|P_1\phi=\phi\} = P_1\HH_0^\perp$, and let $H$ denote the restriction of the central element $C$ to this invariant subspace, i.e. $H=C|_{\mathcal K}=P_1P_2P_1|_{\mathcal K}$. Then we define
\begin{equation}\label{V2H}
    V:\HH_0^\perp\to\KK\otimes \Cx^2
\end{equation}
by
\begin{equation}\label{Vphi}
    V\phi=iP_1\phi\otimes {\ket+}_1+\frac1{\sqrt{H(\idty_{\mathcal K}-H)}}P_1P_2(1-P_1)\phi\otimes{\ket-}_1,
\end{equation}
where ${\ket\pm}_1=\tfrac{1}{\sqrt{2}}\left(\begin{smallmatrix} 1\\ \pm 1\end{smallmatrix}\right)$. One can readily check that this map is unitary. Operators on $\KK\otimes\Cx^2$ can be considered as $\BB(\KK)$-valued $2\times2$-matrices; this gives
\begin{align}\label{paulirel}
A_1|_{\HH_0^\perp}&\cong I\otimes \sigma_1,\nonumber\\
A_2|_{\HH_0^\perp} &\cong \alpha(H)\otimes \sigma_1+\beta(H)\otimes\sigma_2,\\ \nonumber
A_3|_{\HH_0^\perp} &\cong \beta(H)\otimes \sigma_3,
\end{align}
where $\alpha,\beta:[0,1]\to \Rl$ are given by $\alpha(h) = 2h-1$, $\beta(h) = 2\sqrt{h(1-h)}$.
The point of this decomposition is, of course, that in the matrix entries we only have functions of the central element $C$ or rather its compression $H$.

\subsection{Attained maximal violations}
We now use the detailed form (\ref{paulirel}) of the operators to get better information about the states where large violation is attained. Both algebras now have a central element, $C_A$ and $C_B$, respectively, giving the compressions $H_A$ and $H_B$. When these are fixed $h$-numbers, the four operators $A_i,B_j$ are completely fixed $2\times2$-matrices, and we can explicitly find a state on $\Cx^2\otimes\Cx^2$ maximizing the violation. In general we can do this maximization at every pair of values, which by \eqref{maxviol} gives the function
\begin{equation}\label{betaCaCb}
    \beta(h_A,h_B)=2\sqrt{1+4\sqrt{h_A(1-h_A)}\,\sqrt{h_B(1-h_B)}}
\end{equation}
plotted in Fig.~\ref{fig:bump}.
\begin{figure}[h]
  \centering
  \includegraphics[width=8cm]{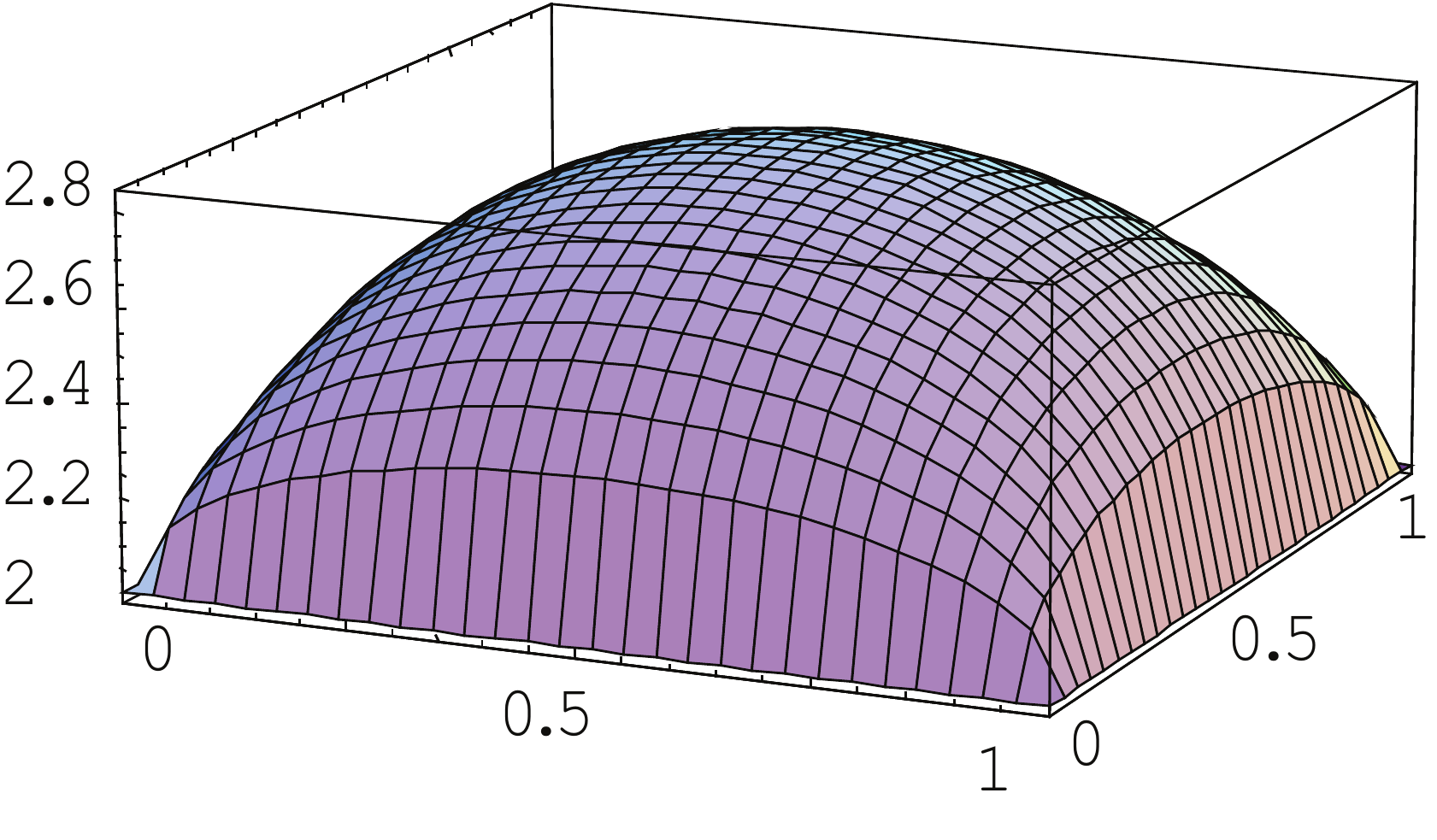}
  \caption{Maximal CHSH-correlation as a function of the central parameters of Alice and Bob}
 \label{fig:bump}
 \end{figure}
Given the joint probability distribution of $H_A$ and $H_B$, the largest attainable CHSH-correlation will be the expectation of (\ref{betaCaCb}) with respect to this distribution. Obviously, for large correlation we want to choose a joint distribution which is concentrated as near the point $h_A=h_B=\frac12$ as possible.

This leads to the following three cases:
\begin{enumerate}
\item When $\frac12$ is an eigenvalue of both $H_A$ and $H_B$, we will chose a maximizing vector from these eigenspaces. Then the CHSH-correlation will be exactly $2\sqrt2$.
\item  When $\frac12$ lies in both spectra, but for one of these operators is not an eigenvalue (i.e., lies in the continuous spectrum), the CHSH-correlation for any state represented by a density operator in the given Hilbert space will be strictly less than $2\sqrt2$, but can be chosen arbitrarily close to this value.
\item Finally if $\frac12$ is not in the spectrum of $H_A$ or $H_B$, the CHSH-correlation is less than $2\sqrt2-\varepsilon$ for some $\varepsilon>0$ and all states.
\end{enumerate}

To characterize the structure of the maximally violating states, we would now like to extract from the case 1 as much information as possible about further expectation values, including those not directly measured in a Bell experiment. Similarly, in case 2 we are interested in the limits of expectation values $\tr(\rho_nA)$, for $\rho_n$ a sequence of density operators with asymptotically maximal violation. It is convenient to treat these two cases on the same footing by choosing a convergent subsequence of the $\rho_n$ in the weak*-topology, and thereby find an exactly maximally violating limiting state. This is possible if we extend the notion of ``states'' from density operators to arbitrary expectation value functionals  $\omega:\mathcal{B}(\HH_A\otimes \HH_B)\to\Cx$. Of these we only require linearity, positivity and normalization, so they are states in the sense of C*-algebra theory. The state space of a C*-algebra is compact with respect to ``convergence of all expectation values'', so convergent subsequences in this wider setting always exist. Of course, in case 2 a sequence $\rho_n$ can converge only to a {\it singular state} and not a ``normal'' one, given by a density operator. The singular state is never unique, because fixing such a state is the non-commutative analog of explicitly defining a free ultrafilter. However, as will be seen below, all these states may well agree on some observables of interest. 

Geometrically, the CHSH expression with fixed $A_i,B_j$ is an affine functional on the state space of $\mathcal{B}(\HH_A\otimes \HH_B)$. It reaches its maximum at an extreme point, so it is not surprising that this entails some special relations. The typical tool here is the Cauchy-Schwarz inequality in the form
$|\omega(X^*Y)|^2\leq \omega(X^*X)\omega(Y^*Y)$. Hence if we know that the expectation of a positive operator like $X^*X$ vanishes we can conclude that $\omega(X^*Y)=\omega(Y^*X)=0$ for all $Y$. This approach has been applied \cite{SW871III} to the CHSH inequality by writing $2\sqrt2-T$ as the sum of several operators of the form $X^*X$. Here we can achieve similar results by looking at the explicit form (\ref{paulirel}) of the operators $A_i,B_j$ after the transformation (\ref{Vphi}).

Under this transformation Alice's Hilbert space becomes $\HH_0^A\oplus(\KK_A\otimes\Cx^2)$, so that
$H_A$ is an operator on $\KK_A$, and the given operators $A_1,A_2$ take the form (\ref{paulirel}). Of course, an analogous decomposition holds for Bob. The projections onto these subspaces as well as the spectral families of the operators $H_A,H_B$ commute with all $A_i,B_j$. Suppose we take a joint spectral projection $P_\varepsilon$ of the commuting operators $H_A,H_B$ corresponding to a set of distance $\varepsilon$ to the point $(\frac12,\frac12)$. Then $TP_\varepsilon$ is strictly smaller than $2\sqrt2$, so a maximizing state must have $\omega(P_\varepsilon)=0$. Hence a maximally violating state must vanish on $\HH_0^A\otimes\HH_B$ and $\HH_A\otimes\HH_0^B$, and its restriction to $\mathcal{B}(\KK_A\otimes\KK_B)$ must be a state $\deltahalf$ giving the sharp values $\frac12$ to both $H_A$ and $H_B$, in the sense that
\begin{equation}\label{delhalf2}
    \deltahalf\Bigl((H_A-\frac12\idty)^2\Bigr)
    =\deltahalf\Bigl((H_B-\frac12\idty)^2\Bigr)
    =0.
\end{equation}
At such a point we can set $\alpha(h)=0$ and $\beta(h)=1$ in formula \eqref{paulirel}, and its analogue for $B_1,B_2$, and just consider the maximization of the CHSH expression with fixed qubit operators $A_i=B_i=\sigma_i$. Since the maximum of $T$ for these operators is attained at the unique pure state
\begin{equation}\label{violatingvector0}
\Psi_0=\frac{1}{\sqrt{2}}(e^{-i\pi/4} |++\rangle+|--\rangle)\in \Cx^2\otimes \Cx^2 = \Cx^4
\end{equation}
(where $|\pm\rangle$ are the eigenvectors of $\sigma_3$), we conclude that the overall state must be of the form
\begin{equation}\label{violstate}
    \omega=\deltahalf\otimes |\Psi_0\rangle\langle\Psi_0|.
\end{equation}
It is clear, that conversely, any state of this description will be maximally violating.

In case 1, the explicit form for the maximally violating pure states in the original representation is now
\begin{equation}\label{violatingvector2}
\Psi = \frac 1{\sqrt 2} [e^{-i\pi/4}e^{A,+}\otimes e^{B,+}+e^{A,-}\otimes e^{B,-}],
\end{equation}
where e.g.
\begin{equation}\label{eeq}
e^{A,\pm}:=V_A^*(f_A\otimes |\pm\rangle) = (-iP^A_1\pm (\idty-P^A_1))(\sqrt{2}P^A_2f),
\end{equation}
and $f_A$ is a normalized eigenvector of $H_A$ belonging to the eigenvalue $\frac 12$. In case 2, we choose a sequence $(f^A_n)$ of unit vectors such that $\|H_Af^A_n -\tfrac 12 f^A_n\|\rightarrow 0$ (usually called \emph{approximate eigenvectors}). Defining $\Psi_n$ using $f_n^A$ and $f_n^B$ as in \eqref{violatingvector2}, we get the asymptotic maximal violation $\lim_{n\rightarrow \infty} \langle \Psi_n|T\Psi_n\rangle = 2\sqrt 2$; the states $|\Psi_n\rangle\langle \Psi_n|$ approximate some singular state of the form \eqref{violstate}. This systematic construction of (approximate) maximally violating wave functions will be used in the next section.

An interesting corollary of the above structure is the {\it cryptographic security} of maximal CHSH correlations. We are then interested in the possible correlations between the observed data and the measurements made by an eavesdropper ``Eve'' in a separate lab. The measurement of Eve is then described by an operator $E$ commuting with all the operators $A_1,A_2,B_1,B_2$ used by Alice and Bob. Hence $E$ lives on the tensor factor $\KK_A\otimes\KK_B$, and from the form (\ref{violstate}) of the state, it is clear that Eve's results are independent of Alice's and Bob's. 
 This is summarized in the following proposition.

\begin{prop} \label{properties} Let $\omega$ be a state maximally violating the CHSH inequality on operators $A_1,A_2,B_1,B_2$. Let $p$ be a non-commutative polynomial in four variables, and set $P=p(A_1,A_2,B_1,B_2)$
and $P^0=p(A_1^0,A_2^0,B_1^0,B_2^0)$, where $A_i^0=\sigma_i\otimes\idty$ and $B_i^0=\idty\otimes\sigma_i$.
Then for any operator $E$ commuting with all $A_i,B_j$:
\begin{equation}\label{secure}
    \omega(EP)=\omega(E)\langle \Psi_0 |P^0\Psi_0\rangle.
\end{equation}
\end{prop}

\section{Position measurements at different times}
\label{sec:position}

We now proceed to the case of position measurements. The Heisenberg picture position operator of a massive, freely evolving nonrelativistic particle with mass $m$ is given by
\begin{equation}\label{evolution}
\mathsf{Q}_t = \mathsf{P}t/m+\mathsf{Q}, \,\,\, t\in \Rl,
\end{equation}
where $\mathsf{Q}$ and $\mathsf{P}$ are the standard position and momentum operators, acting in the Hilbert space $L^2(\Rl)$ (in particular, $\mathsf{P}=-i\hbar\frac{d}{dx}$). Concerning measurements of $Q_t$ (position at time $t$), we are only interested in recording whether the outcome lies in a fixed interval $\Delta\subset \Rl$, in which case we assign the value ''$+1$'' to it; otherwise we label it ''$-1$''. The corresponding two-valued observable is $2\chi_\Delta(\mathsf{Q}_t)-1$.

We consider the case where Alice makes measurements with one particle, and Bob with another one; let $A_1:=2\chi_{\Delta_{1}}(\mathsf{Q}^A)-1$ and $A_2:=2\chi_{\Delta_{2}}(\mathsf{Q}^A_{t})-1$ be the position measurements for Alice's particle at time zero and time $t>0$, with intervals $\Delta_{1}$ and $\Delta_{2}$, respectively, and let $B_i$ be the similar ones for Bob. For simplicity, we suppose that both use the same measurement intervals, same time $t$, and particles of same mass $m$. Now we are in a situation discussed in the preceding section.

Since the operators are identical for both parties, we consider only Alice's part and drop the associated index when there is no confusion. The two projections are now $P_1 = \chi_{\Delta_1}(\mathsf{Q})$ and $P_2=\chi_{\Delta_2}(\mathsf{Q}_t)$. We begin with the fact that the pair $(\mathsf{Q},t^{-1}m\mathsf{Q}_t)$ is canonically conjugated, and therefore unitarily equivalent to the pair $(\mathsf{Q},\mathsf{P})$, the unitary operator in question being simply $e^{i(2t)^{-1}m\mathsf{Q}^2/\hbar}$. With this equivalence, $P_1\simeq \chi_{\Delta_1}(\mathsf{Q})$, $P_2\simeq \chi_{t^{-1}m\Delta_2}(\mathsf{P})$; in the following, we will simply replace the $P_i$ with these operators.

The idea is to consider the possibility of maximal violation of the CHSH inequality for three types of concrete choices for the localization intervals $\Delta_i$, exhibiting different commutativity behavior of the position and momentum projections $P_1$ and $P_2$ \cite{Busch}:
\begin{itemize}
\item[(1)] For compact intervals, $P_1$ and $P_2$ are \emph{partially commutative}, i.e. $\ker[P_1,P_2]=\ker A_3$ is neither $\{0\}$ nor $\HH$. Indeed, there are common 0-eigenvectors of $P_1$ and $P_2$.
\item[(2)] For half-lines, the projections are \emph{totally noncommutative}, i.e. $\ker[P_1,P_2] = \{0\}$, and $\mathcal{K}=L^2(\Delta_1)$.
\item[(3)] For periodic sets, the periods can be chosen so that $P_1$ and $P_2$ are commutative, i.e.
$\ker [P_1,P_2] = \HH$. Then $\mathcal{K}=\{0\}$.
\end{itemize}
(The full characterization of commuting functions of $\mathsf{Q}$ and
$\mathsf{P}$ is given in \cite{BuschII}; for a generalization to Abelian groups, see \cite{Ylinen}.) The projections apparently depend on various parameters $\Delta_1,\Delta_2, t,m$; however, as the dilations are represented by unitary operators, the only relevant parameter is the scale of the $\mathsf{Q}$-interval relative to the $\mathsf{P}$-interval. In case (2) there is no specific scale, because the projections are invariant under dilations; hence the structure of the Bell inequality violations does not depend at all on the parameters. In cases (1) and (3), $\mathsf{Q}$- and $\mathsf{P}$-sets are characterized by lengths $l_1$ and $mt^{-1}l_2$, respectively, where the $l_i$ are proportional to the lengths (case (1)) or periods (case (2)) of the sets $\Delta_i$. If we fix the unit of position as $l_1$ (thereby making the position variable dimensionless), the unit of momentum will be $\hbar l_1^{-1}$; in these units, the above characteristic lengths are $1$ and
\begin{equation}\label{ueq0}
u=\frac{ml_1l_2}{t\hbar},
\end{equation}
respectively. We can equally well fix the unit of momentum as $mt^{-1}l_2$, in which case the unit of position is $t\hbar /(ml_2)$; the characteristic lengths are then $u$ and $1$, respectively. Hence the only relevant parameter is the dimensionless scale $u$. For technical reasons, we will use the first choice of units in case (1) and the second in (2). For both choices of units, the associated operators are dimensionless; we will denote these by $Q$ and $P$.

\subsection{Compact intervals: partially commutative case}
Here we let $\Delta_i\neq \emptyset$ be a compact interval for $i=1,2$. As already mentioned, $\ker A_3$ is nontrivial; however, $P_1(\HH)=L^2(\Delta_1)\subset \HH_0^{\perp}$ (see e.g. \cite{Busch}),
so the relevant subspace $\mathcal{K}$ is just $L^2(\Delta_1)$.
It is convenient to choose the length scales as $l_i:=d_i/2$, with $d_i$ the length of $\Delta_i$; passing to the units where $l_1$ is $1$ as discussed above, we see that the relevant operator $H$ is unitarily equivalent to
$$H_u := \chi_{[-1,1]}(Q)\chi_{[-u,u]}(P)\chi_{[-1,1]}(Q)\in \BB\big(L^2([-1,1])\big),$$
where $u$ is given by \eqref{ueq0}, i.e. $u= md_1d_2/(4t\hbar)$.
(This equivalence can be seen easily by first applying the usual translation and ''velocity boost'' unitaries with appropriate shifts to center the intervals to the origin,
and then dilating by $d_1/2$.)

The structure of $H_u$ has been extensively studied because of its relevance in band- and timelimiting of signals (see, for instance \cite[pp. 21-23]{Daubechies}, \cite[pp. 121-132]{Dym}, \cite{Slepian}, or the original papers by Landau, Pollack and Slepian \cite{Pollack, LandauI, Slepian}). We summarize the relevant mathematical facts briefly in the following paragraph.

The operator $H_u$ is explicitly given by
$$(H_u\varphi)(v) = \int_{-1}^1 \frac{\sin(u(v-w))}{\pi (v-w)}\, \varphi(w)\,dw, \, \, \varphi\in L^2([-1,1]),$$
from which it follows that $H_u$ commutes with the differential operator
$\frac{d}{dv} \left[(1-v^2)\frac{d}{dv}\right]-u^2v^2$ that determines the
angular part of the wave equation in prolate spheroidal coordinates. This differential operator has
a complete orthonormal set of eigenfunctions $\psi_n^u\in L^2([-1,1])$, $n=0,1,\ldots$, called \emph{angular prolate spheroidal wave functions}.  In
the notation of \cite{Robin} we have $\psi_n^u(v)=\sqrt{n +\tfrac 12}{\rm Ps}_{n}(v,u)$. The corresponding eigenvalues
$\lambda_n(u)$ of $H_u$ are
\begin{equation}\label{eigenvalue}
\lambda_n(u)= 2u\pi^{-1} S_{n}^{(1)}(1,u)^2\in (0,1), \,\,\, n=0,1,2,\ldots,
\end{equation}
where $S_{n}^{(1)}(\cdot,u)$ is the \emph{radial} prolate spheroidal wave function of the first kind.
In particular, $\lambda_n(u)$ depends continuously on $u$. In addition, we have $1>\lambda_n(u)>\lambda_{n+1}(u)>0$ for all $n$ and $u$.

Now $\tfrac 12\in \sigma(H)$ exactly when $u$ is chosen so as to make $\lambda_{n}(u)=\tfrac 12$ for some $n$.
Since $\lim_{u\rightarrow\infty} \lambda_n(u) =1$, and $\lim_{u\rightarrow\infty} \lambda_n(u) =0$
for fixed $n$ (see \cite{Slepian}), it follows by continuity that for each $n$ we get at least one value $u_n\in (0,1)$ with $\lambda_n(u_n) = \tfrac 12$. On the other hand, $H_u\leq H_{u'}$ if $u\leq u'$, so each $\lambda_n(u)$ is an increasing function of $u$, and $u_n$ is thus uniquely determined. Since for given $n$, we have $\lambda_n(u)> \lambda_{n+1}(u)$ for all $u$ it follows by continuity that $u_n<u_{n+1}$, i.e. the sequence $(u_n)$ is increasing.

Figures \ref{fig:eigvs} and \ref{fig:norm} show the $u$-dependence of the largest eigenvalues, as well as the relevant commutator norm $\|A_3\|$, obtained from the above representation. The critical values $u_n$
can be computed numerically; the smallest two are approximately $u_0\approx 0.849$ and $u_1\approx 2.381$.
\begin{figure}[h]
  \centering
  \includegraphics[width=8cm]{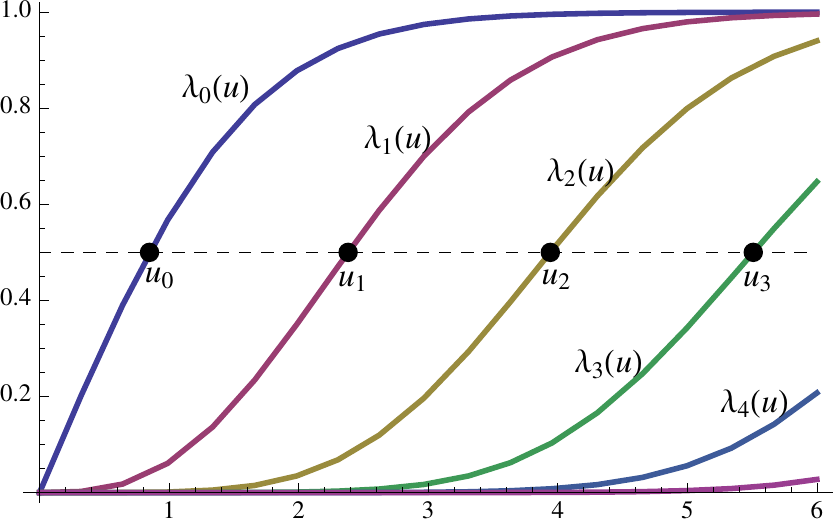}
  \caption{The eigenvalues $\lambda_n(u)$ of $H$ as functions of the parameter $u=\tfrac 14(t\hbar)^{-1}md_1d_2$}
 \label{fig:eigvs}
 \end{figure}
\begin{figure}[ht]
  \centering
  \includegraphics[width=8cm]{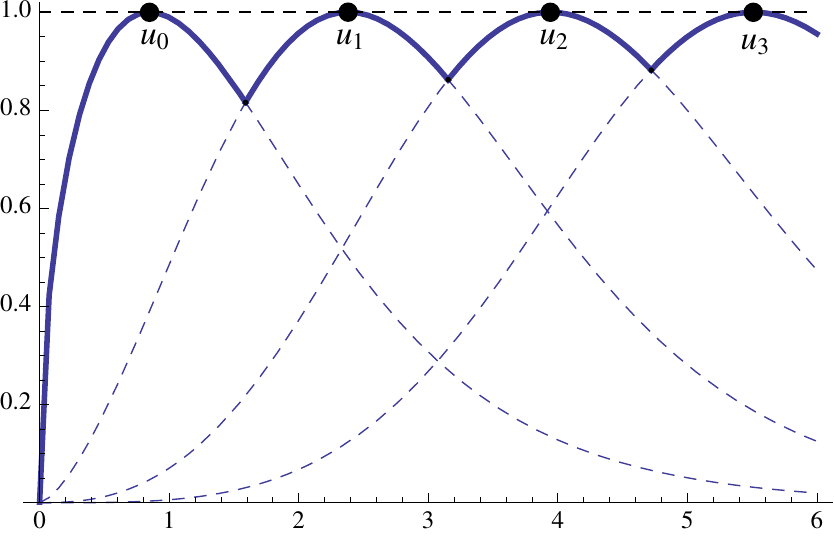}
  \caption{The norm of the operator $A_3$ as a function of $u$}
 \label{fig:norm}
 \end{figure}

Summarizing, a state that maximizes the violation of CHSH inequality for position measurements exists if and only if both Alice and Bob adjust their parameters in such a way that \eqref{ueq0} holds with $u$ one of the critical values, say $u_n$ for both Alice and Bob \footnote{In order to simplify the form of the maximally violating wave function, we have assumed in the beginning that Alice and Bob use the same parameter values. Clearly, they could just as well use different values, in which case Alice must have $u_A=u_n$ for some $n$, and Bob can have $u_B=u_{n'}$ for some other number $n'$.}. Using \eqref{violatingvector2}, the corresponding wave function $\Psi\in L^2(\Rl^2,dq_A,dq_B)$ can then be expressed quite explicitly. For this we need the functions $e^\pm$ of \eqref{eeq}; but now
$\sqrt{2}P_2\psi_n^{u_n} = \sqrt{\tfrac 12(n+\tfrac 12)}{\rm Ps}_n(\cdot,u_n)\in L^2(\Rl)$, so $e^\pm(q)$ is simply this spherical function, multiplied with $-i$ if $q\in \Delta_1$, and with $\pm 1$ otherwise. In particular, the wave function $\Psi$ is discontinuous at the lines $q_A=\pm 1$ and $q_B=\pm 1$. In the case where the intervals are centered at the origin, i.e. $\Delta_{i}= [-d_{i}/2,d_{i}/2]$, we get

\begin{equation}\label{concretevector}
\Psi(q_A,q_B) = C_0(q_A,q_B)\, \frac{1}{d_{1}}(n+\frac 12){\rm Ps}_n\big(2q_A/d_{1},u_n\big){\rm Ps}_n\big(2q_B/d_{1},u_n\big) e^{i\Theta(q_A,q_B)},
\end{equation}
where
\begin{equation}\label{constant}
C_0(q_A,q_B) = \begin{cases} \sqrt{1+1/\sqrt{2}}, & q_A,q_B\in \Delta_1, \text{or } q_A,q_B\notin \Delta_1\\
\sqrt{1-1/\sqrt{2}}, & \text{otherwise}\end{cases}
\end{equation}
is real, and the phase $\Theta(q_A,q_B)$ is given by
\begin{eqnarray}
\Theta(q_A,q_B) &=& -2u_n(q_A^2+q_B^2)/(d_1d_2) +\begin{cases} \phi_0^++\pi, & q_A,q_B\in \Delta_1\\
               \phi_0^{+}, & q_A,q_B\notin \Delta_1,\\
               \phi_0^-+\pi/2, & \text{otherwise}.\end{cases}\\
\phi_0^\pm &=& \mp \arctan((\sqrt{2}\pm 1)^{-1})\nonumber.
 \end{eqnarray}
Figure \ref{fig:wave} shows the picture of the simplest choice for the wave function.
\begin{figure}[ht]
  \centering
  \includegraphics[width=13cm]{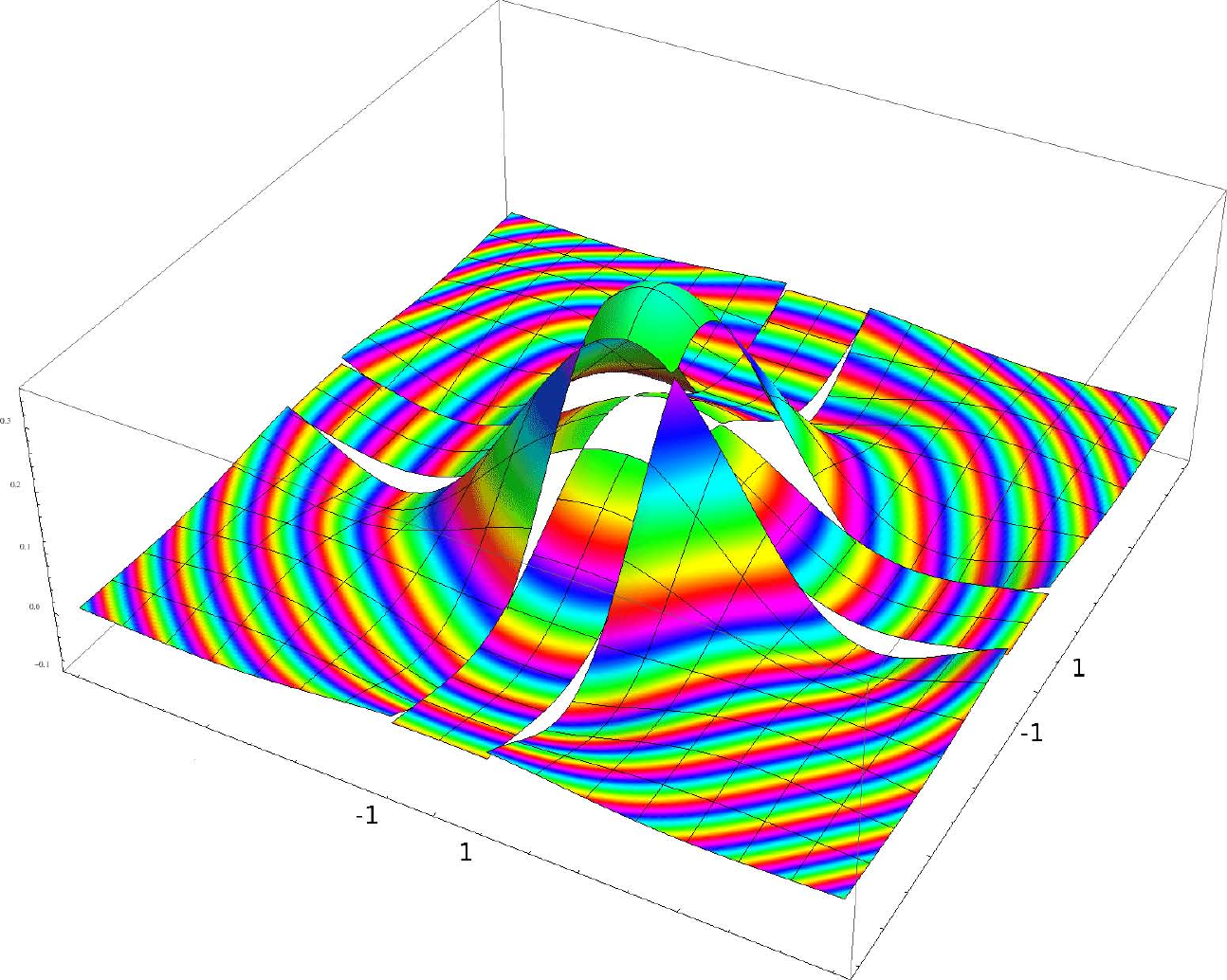}
  \caption{The maximally violating wavefunction \eqref{concretevector} with $n=0$ ($u=u_0\approx 0.849$), $\Delta_1 = \Delta_2 = [-1,1]$. The plotted function is $\Psi(q_A,q_B)$ \emph{without} the complex phase factor $e^{i\Theta(q_A,q_B)}$; the phase $\Theta(q_A,q_B)$ is shown in color, going through red, yellow, green, cyan, blue, magenta, and again red, as the value varies from $0$ to $2\pi$. Note the discontinuity lines $q_A=\pm 1$, $q_B=\pm 1$ of the factor $C_0(q_A,q_B)$ (see equation \eqref{constant}), marking the measurement interval.}
 \label{fig:wave}
 \end{figure}

\subsection{Half-lines: totally noncommutative case}\label{halfline}

Here we set $\Delta_{i} =[x_i,\infty)$, $i=1,2$, $x_i\in \Rl$; then $\ker A_3=\{0\}$, and we have the tensor product representation for the full operators. As before, $\mathcal{K}=L^2(\Delta_1)$. We first apply the unitary shifts that transform the situation to the dilation invariant case $\Delta_i = [0,\infty)$, $i=1,2$; then the units of position and momentum play no role, and $A_1 \cong {\rm Sign}(Q)$ and $A_2 \cong {\rm Sign}(P)$. In particular, the spectrum of $H$ does not depend on any of the parameters.

The transformation $V:L^2(\Rl)\to L^2([0,\infty))\otimes \Cx^2$ given by \eqref{Vphi} has now a particularly simple form, namely
\begin{equation}\label{Vconcrete}
V\phi = i[\chi_{[0,\infty)}(Q)\phi \otimes |+\rangle_1 +\Pi\chi_{(-\infty,0]}(Q)\phi\otimes |-\rangle_1],
\end{equation}
where $\Pi:L^2((-\infty,0])\to L^2([0,\infty))$ is the parity operator. However, we still have to determine the spectrum of the operator $H = \chi_{[0,\infty)}(Q)\chi_{[0,\infty)}(P)\chi_{[0,\infty)}(Q)$, acting on $\mathcal{K} = L^2([0,\infty))$. Here it is convenient to utilize the dilation invariance of the projections; we seek a unitary operator $W$ that diagonalizes the dilation generator $\mathsf{D}:=\mathsf{Q}\mathsf{P}+\mathsf{P}\mathsf{Q}$ by way of
\begin{equation}\label{dilationtransf}
W\mathsf{D}W^*=Q\otimes I_{\Cx^2}
\end{equation}
Such a unitary is obtained by first separating the positive and negative half-axis by using the $V$ above, then expanding $L^2([0,\infty))$ into the full $L^2(\Rl)$ via the unitary operator
$U_{+}:L^2([0,\infty))\to L^2(\Rl)$, where
\begin{equation}\label{Uplus}
(U_+\varphi)(\lambda) = \sqrt{2}e^{\lambda}\varphi(e^{2\lambda}),
\end{equation}
and then applying the Fourier-Plancherel operator $F$. Indeed, the operator $(U_+\otimes I_{\Cx^2}) V$ transforms $\mathsf{D}$ into $P\otimes \idty_{\Cx^2}$, so we get \eqref{dilationtransf} by setting $W:=-i(FU_+\otimes I_{\Cx^2})V:L^2(\Rl)\to L^2(\Rl)\otimes \Cx^2$ (where the factor $-i$ is chosen for convenience.) This unitary operator approximately diagonalizes $P_1$ and $P_2$ simultaneously, meaning that we get explicit form for the functions of $H$ appearing in \eqref{paulirel} (see Appendix). We can then explicitly compute $H=P_1P_2P_1$: this gives
\begin{equation}\label{Hconnection}
H\cong \frac 12 (\idty+\tanh(\tfrac 12 \pi Q)),
\end{equation}
acting on the space $L^2(\Rl)= U_+L^2([0,\infty))$.

It follows from \eqref{Hconnection} that the spectrum of $H$ is purely absolutely continuous, and contains the point $\tfrac 12$. Hence, maximally violating states $\omega$ exist, are of the form \eqref{violstate}, and each of them is singular. Note that the tensor product representation space is now $L^2(\Rl^2) \otimes \Cx^4$.

It is possible to further specify the properties of the restriction of $\omega$ to the first tensor factor $L^2(\Rl^2)$. According to \eqref{Hconnection}, spectral projections of $H$ associated with intervals around $\tfrac 12$ correspond bijectively to those of $Q$ around $0$. Hence, equations \eqref{delhalf2} imply
that $\delta_{\frac 12,\frac 12} (f(Q_A,Q_B)) = f(0,0)$ for any bounded measurable function $f:\Rl^2\to \Cx$ continuous at $(0,0)$.
With this information, we can now go back to the representation where $T$ is given by \eqref{bellop} with $A_1 = {\rm Sign}(Q_A)$, $A_2 = {\rm Sign}(P_A)$, and $B_i$ similarly; using \eqref{dilationtransf}, as well as the Schwarz inequality, we see that any maximally violating state $\omega$
satisfies
\begin{equation}\label{dilations}
\omega(f(\mathsf{D}_A,\mathsf{D}_B)\,X) = \omega(X\,f(\mathsf{D}_A,\mathsf{D}_B))=f(0,0)\omega(X),
\end{equation}
if $f:\Rl^2\to \Cx$ is a bounded measurable function continuous at $(0,0)$, and $X\in \BB(L^2(\Rl^2))$ is arbitrary. In particular, each maximally violating state is \emph{invariant under dilations} in this representation. Since $\mathsf{Q}$ and $\mathsf{P}$ transform covariant under dilations, this means that \emph{maximally violating states are concentrated on $(0,0)$ and infinity, in both position and momentum representations}; the precise statement is the following observation.
\begin{prop}\label{dilationprop} Let $\omega\in \BB(L^2(\Rl^2))^*$ be a dilation invariant state, and let $f:\Rl^2\to \Cx$ be a continuous function vanishing at the origin and infinity. Then
$$
\omega(f(\mathsf{Q}_A,\mathsf{Q}_B)) =\omega(f(\mathsf{P}_A,\mathsf{P}_B)) = 0.
$$
\end{prop}

\begin{proof}
Set $K(r_1,r_2):=B(r_2)\setminus B(r_1)$, where $B(r)$ is the open ball in $\Rl^2$ of radius $r$ centered at the origin. By dilation invariance,
$\omega(\chi_{B(r)}(\mathsf{Q}_A,\mathsf{Q}_B))= \omega(\chi_{B(1)}(\mathsf{Q}_A,\mathsf{Q}_B))$ for all $r>0$, so for $r_1< r_2$ we get
$\omega(\chi_{K(r_1,r_2)}(\mathsf{Q}_A,\mathsf{Q}_B))= 0$.
From the positivity of $\omega$ it follows that $\omega(f(\mathsf{Q}_A,\mathsf{Q}_B))=0$ for any bounded measurable function $f:\Rl^2\to \Cx$ supported in $K(r_1,r_2)$. Now if $f:\Rl^2\to \Cx$ is continuous and vanishes at  both zero and infinity, then $\lim_{n\rightarrow\infty} \|\chi_{K(1/n,n)}f-f\|_\infty = 0$,
and hence $\omega(f(\mathsf{Q}_A,\mathsf{Q}_B))=0$ by the norm continuity of $\omega$. The case with $\mathsf{P}_A$ is similar.
\end{proof}

We now wish to find wave functions approximating the maximally violating singular states $\omega$. Since any such state is dilation invariant, one can expect that the approximating wave functions in $L^2(\Rl^2)$ would basically look like $1/\sqrt{|xy|}$, but with some regularization at the coordinate axis and infinity. (Here $1/\sqrt{x}$ comes from formally solving the ''eigenvalue equation'' $\mathsf{D}\psi = 0$.)

In order to construct such approximating wave functions, we proceed as described in the preceding section. The approximate eigenvectors of \eqref{Hconnection}, corresponding to the point $\frac 12\in \sigma(H)$ are of the form $g_\epsilon = (2\epsilon)^{-\frac 12} g(x/(2\epsilon))$, where $g\in L^2(\Rl)$ is an arbitrary unit vector, and $\epsilon>0$ is small. Hence, the corresponding vectors for the original $H = \chi_{[0,\infty)}(Q)\chi_{[0,\infty)}(P)\chi_{[0,\infty)}(Q)$, acting on $L^2([0,\infty))$, are $f_\epsilon:=iU_+^*F^*g_\epsilon$; explicitly, they are of the form
$$
f_\epsilon(x) = i\sqrt{\frac{\epsilon}{x}}f(\epsilon \ln x), \, \, x>0,
$$
where $f\in L^2(\Rl)$ is an arbitrary unit vector. Hence, the wave functions we are seeking are given by \eqref{violatingvector2}, with $e_\epsilon^\pm = V^*(f_\epsilon\otimes |\pm\rangle)$ (for both Alice and Bob). These can now be obtained from \eqref{Vconcrete}:
\begin{align*}
e_\epsilon^+(x) &= \frac{1}{\sqrt{2}} f_\epsilon(|x|), & e_\epsilon^-(x) = \frac{1}{\sqrt{2}}{\rm Sign}(x)f_\epsilon(|x|).
\end{align*}
Hence,
\begin{eqnarray}
\Psi_\epsilon(q_A,q_B)
    &=& \frac{1}{2\sqrt{2}} (e^{-i\pi/4}+{\rm Sign}(q_Aq_B))\
        \frac{\epsilon}{\sqrt{|q_Aq_B|}} f(\epsilon \ln |q_A|)
         f(\epsilon \ln |q_B|).\label{approx}
\end{eqnarray}
This same formula appears in \cite{Auberson}; however, the paper does not seem to contain any systematic derivation for the result.

The approximating wave function in the original representation, where both $A_i$ are position measurements, is then
\begin{eqnarray*}
(q_A,q_B)&\mapsto& \Psi_\epsilon(q_A,q_B) e^{-i\frac 12 \hbar^{-1}mt^{-1}(q_A^2+q_B^2)},
\end{eqnarray*}
where the measurement intervals are $[0,\infty)$ for both time zero and $t$.

We close this subsection by demonstrating that the singular states that can be approximated by the wave functions \eqref{approx} actually depend on the regularizing function $f$, even though they are all maximally violating and dilation invariant. One property of a dilation invariant state $\omega$ that we can easily determine is the expectation value $\omega_0^Q:=\omega(h(Q_A,Q_B))$, where $h:\Rl^2\to\Rl$ is bounded, measurable, and continuous at the origin, with $h(0,0)=1$. By dilation invariance, this does not depend on $h$, and describes the ''weight'' of the state at the origin in the position spectrum. Note that by Proposition \ref{dilationprop}, each dilation invariant state is concentrated at zero and infinity. However, the distribution of weight between these points is not fixed: by direct calculation using \eqref{approx}, we get
\begin{equation}
\omega_0^Q = \lim_{\epsilon\rightarrow 0} \langle \Psi_\epsilon |\chi_{[-a,a]\times [-a,a]}(Q_A,Q_B)\Psi_\epsilon\rangle= \lim_{\epsilon\rightarrow 0} \left(\int_{-\infty}^{\epsilon \ln a} |f(x)|^2\, dx\right)^2
= \left(\int_{-\infty}^{0} |f(x)|^2\, dx\right)^2,
\end{equation}
which may attain any value in $[0,1]$, depending on where $f$ is concentrated.

\subsection{Periodic sets: commutative case}
Here we simply want to make a remark about the commutativity, without trying to analyze the periodic set case systematically. Consider sets of the form
$$\Delta_i = p_i\mathbb{Z}+[0,p_i/2],$$ where $p_i>0$ are the periods.
Now we choose the length scales as $l_1:=p_1$ and $l_2:=p_2/(2\pi)$. Passing to the units where $mt^{-1}l_2$ is $1$ (see the beginning of the section) via the associated dilation, we get $A_1 \simeq g_{u}(Q)$, $A_2 \simeq g_{2\pi}(P)$, with
$$
g_{v}(q) = \begin{cases} 1,& q\in v([0,\tfrac 12 )+\mathbb{Z}),\\ -1, & \text{otherwise}\end{cases}
$$
and $u$ again given by \eqref{ueq0}; $u=mp_1p_2/(2\pi t\hbar)$. The reason for the choice of units is that $A_1$ and $A_2$ (or, equivalently, $P_1$ and $P_2$) commute if $u^{-1}$ is an integer. This can easily be seen by noting that $g_{v}(x) ={\rm Sign}(\sin(2\pi x/v))$, and using the commutation relation for the Weyl operators. Moreover, the converse is also true; see the general characterization of commuting functions of $Q$ and $P$ \cite{BuschII,Ylinen}. In the commuting case the spectrum of $H$ contains only the points $0$ and $1$, and we have $A_3=B_3 = 0$. The CHSH-inequality is then actually satisfied for all states, and the situation is classical.

It is interesting to observe that when the parameter $u^{-1}$ is slightly perturbed from an integer, the commutator norm $\|A_3\|$ discontinuously jumps to a nonzero value, which is large enough to allow a violation of Bell's inequality. In order to show this, we take $u= 1+\epsilon$, with $\epsilon>0$. Now $g_{2\pi}$ has the Fourier expansion
$$
g_{2\pi}(p) =  \sum_{n\in \mathbb{Z}} c_n \frac{1}{\sqrt{2\pi}}e^{inp},
$$
where $c_{-n} = -c_n$, and
$\sum_{n\in \mathbb{Z}} |c_n(u)|^2 =  \int_{-\pi}^\pi |s_u(p)|^2 \, dp= 2\pi$. For each $0<\epsilon<1/2$ choose a unit vector $\psi_\epsilon\in L^2(\mathbb{R})$ with support in $[0,\epsilon)$. Then $g_{1+\epsilon}(Q)\psi_\epsilon = \psi_\epsilon$, and
$$
g_{2\pi}(P)s_{1+\epsilon}(Q))\psi_\epsilon = g_{2\pi}(P)\psi_\epsilon = \sum_{n\in\mathbb{Z}} \frac{1}{\sqrt{2\pi}}c_n e^{inP}\psi_\epsilon,
$$
where the series converges in $L^2(\mathbb{R})$ because the terms are orthogonal. Since $(e^{inP}\psi_\epsilon)(x) = \psi_\epsilon(x+n)$, we get
$$g_{1+\epsilon}(Q)e^{inP}\psi_\epsilon = \begin{cases} e^{inP}\psi_\epsilon, & 0\leq n <\tfrac 12 (1/\epsilon-1),\\
-e^{inP}\psi_\epsilon, & -\tfrac 12 (1+1/\epsilon) < n < 0\end{cases}.$$
It follows that
$$[g_{2\pi}(P),s_{1+\epsilon}(Q)]\psi_\epsilon = 2\sum_{-1/(2\epsilon)<n\leq -1} \frac{1}{\sqrt{2\pi}}c_n e^{inP}\psi_\epsilon+ \varphi_\epsilon,$$
where $\varphi_\epsilon$ is orthogonal to the sum. Hence,
$$|[g_{2\pi}(P),s_{1+\epsilon}(Q)]\|^2\geq \frac{2}{\pi}\sum_{1\leq n < 1/(2\epsilon)} |c_n|^2.$$ Here the right hand side tends to $2$ as $\epsilon \rightarrow 0$, which means that $|[g_{2\pi}(P),gs_{1+\epsilon}(Q)]\|\geq \sqrt{2}$ for all sufficiently small $\epsilon>0$. Hence, $\|A_3\|$ jumps discontinuously from $0$ to some value larger than $1/\sqrt{2}$, corresponding to the Bell correlation $\sqrt{6}$.

\

\noindent {\bf Acknowledgement:} This work was supported by Finnish Cultural Foundation, Academy of Finland, and the EU Integrated Project SCALA.

\section*{Appendix: Almost simultaneous diagonalization of  ${\rm Sign}(Q)$ and ${\rm Sign} (P)$.}

The diagonalization is provided by the unitary operator $W:L^2(\Rl)\to L^2(\Rl)\otimes \Cx^2$ defined in Section \ref{halfline}. This can be written in the form
$$W\phi = (FU_+\otimes S)(\chi_{[0,\infty)}(Q)\phi \otimes |+\rangle +\Pi\chi_{(-\infty,0]}\phi\otimes |-\rangle),$$
where $|\pm\rangle$ are the eigenstates of $\sigma_3$, $S$ is the Hadamard matrix $S = \frac{1}{\sqrt{2}}\left(\begin{smallmatrix}1&1\\ 1&-1\end{smallmatrix}\right)$, and $U_+$ is defined in \eqref{Uplus}.

For ${\rm Sign}(Q)$, we get
$$W{\rm Sign}(Q)W^*=I\otimes \sigma_1$$
by trivial calculation.
The form of $W{\rm Sign}(P)W^*$ is not so obvious, and the following computation
actually describes how to find a suitable $W$. Let $\varphi\in L^2(\Rl)$ be a Schwartz space function, with $\varphi(0)=0$. The set of such functions is dense in $L^2(\Rl)$. By using
dominated convergence twice, and then Fubini's theorem (noting that $\int_{-\infty}^\infty |\varphi(x)|/|x|\, dx <\infty$) we get
\begin{equation*}
\langle \varphi |{\rm Sign}(P)\varphi\rangle = \lim_{\delta\rightarrow 0+}\sum_{\epsilon_1,\epsilon_2=\pm 1}\int_{\Rl^2_{\epsilon_1,\epsilon_2}} \overline{\varphi(x)}K^\delta_{\epsilon_1,\epsilon_2}(x,y)\varphi(y)\, dxdy,
\end{equation*}
with
$$
K^\delta_{\epsilon_1,\epsilon_2}(x,y) := \frac{1}{2\pi}\left[\frac{1}{\delta\epsilon_1 x-i(x-y)} -\frac{1}{\delta\epsilon_2 y+i(x-y)}\right],
$$
where $\Rl^2_{\epsilon_1,\epsilon_2}$ denotes the appropriate quadrant.
Each kernel $K^\delta_{\epsilon_1,\epsilon_2}$ is invariant under dilations, i.e.
$aK^\delta_{\epsilon_1,\epsilon_2}(x,y) =K^\delta_{\epsilon_1,\epsilon_2}(a^{-1}x,a^{-1}y)$ for all
$a>0$. Using this we can transform them into convolution kernels, and then diagonalize using the
Fourier transform. Indeed, put
$\tilde{K}^\delta_{\epsilon_1,\epsilon_2}(\lambda) := 2K^\delta_{\epsilon_1,\epsilon_2}(\epsilon_1e^{\lambda},\epsilon_2e^{-\lambda})$, $\lambda\in \Rl$;
this gives
$$\tilde{K}^\delta_{\epsilon_1,\epsilon_2}(\lambda) = \frac{1}{\pi}\left[\frac{1}{\delta e^\lambda-2iG_{\epsilon_1,\epsilon_2}(\lambda)}-\frac{1}{\delta e^{-\lambda} +2iG_{\epsilon_1,\epsilon_2}(\lambda)}\right],
$$
where $G_{\pm\pm} = \pm \sinh\lambda$, and $G_{\pm\mp} = \pm \cosh\lambda$.
For $\delta>0$, each $\tilde{K}^\delta_{\epsilon_1,\epsilon_2}$ is both integrable and square integrable,
so we can put
$\hat{K}^\delta_{\epsilon_1,\epsilon_2}(\eta):= \int_\Rl e^{-i\eta\lambda} \tilde{K}^\delta_{\epsilon_1,\epsilon_2}(\lambda)\, d\lambda$. Then we compute
\begin{eqnarray*}
&& \int_{\Rl^2_{\epsilon_1,\epsilon_2}} \overline{\varphi(x)}K^\delta_{\epsilon_1,\epsilon_2}(x,y)\varphi(y)\, dxdy= \int_\Rl d\lambda\, \overline{(V_{\epsilon_1}\varphi)(\lambda)}(\tilde{K}^\delta_{\epsilon_1,\epsilon_2}*V_{\epsilon_2}\varphi)(\lambda)\\
&=& \int_\Rl d\eta\, \overline{(FV_{\epsilon_1}\varphi)(\eta)}\hat{K}^\delta_{\epsilon_1,\epsilon_2}(\eta)(V_{\epsilon_2}\varphi)(\eta)= \langle FV_{\epsilon_1}\varphi|\hat{K}^\delta_{\epsilon_1,\epsilon_2}FV_{\epsilon_2}\varphi\rangle.
\end{eqnarray*}
It remains to take the limit $\delta\rightarrow 0+$. To this end, first note that since e.g.
$
\left|\delta e^\lambda-2iG_{\pm,\mp}(\lambda)\right|^{-1}\leq (\cosh\lambda)^{-1}
$
for any $\delta>0$, and since $\lambda\mapsto 1/\cosh\lambda$ is integrable, it follows that $\hat{K}^\delta_{\pm,\mp}$ are bounded uniformly for $\delta$, and
$\hat{K}^0_{\pm,\mp}(\eta)$ exists with
$\hat{K}^0_{\pm,\mp}(\eta) =  \lim_{\epsilon\rightarrow 0+} \hat{K}_{\pm,\mp}^\epsilon(\eta)$
pointwise by dominated convergence. Hence, the corresponding bounded multiplication
operators on $L^2(\Rl,d\eta)$ converge in the strong operator topology, giving
$$
\lim_{\delta\rightarrow 0+} \langle FV_{\epsilon_1}\varphi|\hat{K}^\delta_{\pm,\mp}V_{\epsilon_2}\varphi\rangle
=\langle FV_{\epsilon_1}\varphi|\hat{K}^0_{\pm,\mp}FV_{\epsilon_2}\varphi\rangle.
$$
Since $\tilde{K}^0_{\pm,\mp}(\lambda) = \mp (i\pi\cosh\lambda)^{-1}$ are even functions, we have
$$
\hat{K}^0_{\pm,\mp}(\eta) = \mp\frac{2}{i\pi} \int_0^\infty \frac{\cos(\eta\lambda)}{\cosh\lambda}\, d\lambda = \pm \frac{i}{\cosh(\eta\pi/2)}.
$$
For diagonal elements, $\tilde{K}^0_{\pm,\pm}= \mp (i\pi\sinh \lambda)^{-1}$, and the corresponding Fourier integral does not exist. However, the limit $\hat{K}^0_{\pm,\pm}(\eta):= \lim_{\delta\rightarrow 0+} \hat{K}^\delta_{\pm,\pm}(\eta)$ exists pointwise, because $\tilde{K}^\delta_{\pm,\pm}$ is an odd function; in fact,
\begin{equation}
\hat{K}^\delta_{\pm,\pm}(\eta) = -2i\int_0^\infty \sin(\eta\lambda)\tilde{K}^\delta_{\pm,\pm}(\lambda)\, d\lambda
  \longrightarrow  \pm\frac{2}{\pi}\int_0^\infty \frac{\sin(\eta\lambda)}{\sinh\lambda}\, d\lambda = \pm\tanh(\frac{\eta\pi}{2}),
\end{equation}
as $\delta\rightarrow 0+$, the singularity at the origin being canceled by the factor $\sin(\eta\lambda)$. We can now use e.g. the bound $|\hat{K}^\delta_{\pm,\pm}(\eta)|\leq M|\eta|\leq M(\eta^2+1)$,
where $M=\tfrac{2}{\pi}\int_0^\infty \lambda(\sinh \lambda)^{-1}\, d\lambda<\infty$, and the fact that $\|\mathsf{Q} FV_\pm\varphi \|=\|\mathsf{P}V_\pm\varphi\|\leq \|\varphi\|+2\|\mathsf{Q}\varphi\|<\infty$, to conclude that
$$
\lim_{\delta\rightarrow 0+} \langle FV_{\epsilon_1}\varphi|\hat{K}^\delta_{\pm,\pm}FV_{\epsilon_2}\varphi\rangle
=\langle FV_{\epsilon_1}\varphi|\hat{K}^0_{\pm,\pm}FV_{\epsilon_2}\varphi\rangle.
$$
The coefficient matrix is thus $\hat{K}^0(\eta) = \tanh(\eta\pi/2)\sigma_3 -{\rm sech}(\eta\pi/2)\sigma_2$.
Finally, taking into account the Hadamard matrix $S$ in the definition of $W$, we get the result
$$W{\rm Sign}(P)W^* = \tanh(Q\pi/2) \otimes \sigma_1+{\rm sech}(Q\pi/2)\otimes \sigma_2.$$

\bibliographystyle{abbrv}

\end{document}